\documentclass[12pt]{article}

\usepackage{ucs}
\usepackage[cp1251,utf8x]{inputenc}
\usepackage[T2A]{fontenc}
\usepackage[russian,romanian,english]{babel}

\usepackage{amsmath}
\usepackage{amssymb}
\usepackage{amsthm}

\usepackage{tabularx}
\usepackage{array,longtable}
\newtheorem{prop}{Proposition}

\newtheorem{lem}{Lemma}
\newtheorem{thm}{Theorem}

\newtheorem{ilem}[lem]{Lemma}

\newcommand{\lmbb}{\ensuremath{L\mathfrak{B}_2{}}}
\newcommand{\set}[1]{\ensuremath{\left\{#1\right\}}}
\newcommand{\seq}[1]{\ensuremath{\left<#1\right>}}
\newcommand{\mbb}{\ensuremath{\mathfrak{B}_2{}}}
\newcommand{\sta}[3]{\ensuremath{{#1}_{#2},\dots,{#1}_{#3}}}
\newcommand{\stb}[2]{\ensuremath{{#1},\dots,{#2}}}
\newcommand{\stc}[4]{\ensuremath{{#1}_{#2}#4,\dots,{#1}_{#3}#4}}
\newcommand{\ccaa}[2]{{{\displaystyle #1}\atop{\displaystyle #2}}}

\usepackage{bbold} % $\mathbb{1}$, $\mathbb{0}$
\usepackage{dsfont} % \mathds{1}
\usepackage{cmll}    % \with = \with

\title{On sufficient conditions for expressibility of constants in the 4-valued extension of the propositional provability logic $GL$}

\author{Andrei RUSU\\
Information Society Development Institute \\
Academy of Sciences of Moldova\\
andrei.rusu@idsi.md}

\date{\ }

\begin{document}

\maketitle{}

\begin{center}
\emph{In memory of Professor Mefodie Rață}
\end{center}
\begin{abstract}
In the present paper we consider the simplest non-classical extension $GL4$ of the well-known propositional provability logic $GL$ together with the notion of expressibility of formulas in a logic proposed by A.~V.~Kuznetsov. Conditions for expressibility of constants in $GL4$ are found out, which were first announced in a author's paper in 1996.  
\end{abstract}

\renewcommand{\baselinestretch}{1.0}\normalsize%
\section{Introduction}
\renewcommand{\baselinestretch}{1.33}\normalsize%
The criteria of completeness with respect to expressibility is well-known in the case of boolean functions \cite{Post1921, Post1941}. A.~V.~Kuznetsov \cite{Kuznetsov1965, Kuznetsov1971} has specified the notion of expressibility to the case of formulas in logical calculi, using the rule of replacement by its equivalent in the given logic. Professor Mefodie Rață has obtained the criterion of completeness relativ to expressibility in propositional intuitionistic logic and its extensions \cite{Ratsa1971, Ratsa1982}.  

We consider the simplest non-classical 4-valued extension of the propositional provability logic of G\"odel-L\"ob $GL$ \cite{Solovay1975} and found out the sufficient conditions for expressibility of constant formulas of this logic. 

\renewcommand{\baselinestretch}{1.0}\normalsize%
\section{Definitions and notations}
\renewcommand{\baselinestretch}{1.33}\normalsize%

{\bf Propositional provability logic $GL$.}
The formulas of the propositional provability calculus of $GL$ are built from the symbols of propositional variables $p, q, r, \dots$ (may be also indexed), by means of the symbols of logical connectives $\&, \vee, \supset, \neg$ and $\Delta$ (represent the unary modal operation of provability by G\"odel), and parentheses. For example, the expressions $(p\&\neg p)$, $(p\supset p)$, $(\Delta(p\&\neg p))$ and $(\neg(\Delta(p\&\neg p)))$ are formulas in the calculus of $GL$, representing the constant formulas denoted in the following by $0, 1, \sigma, \rho$, and we denote the formulas $(p\&\Delta p)$ and $((p\supset q)\&(q\supset p))$ as $\square p$ (box $p$) and $(p\sim q)$ (equivalence of $p$ and $q$). External parentheses are usually omitted. The calculus of the $GL$ is determined by the axioms of the classical calculus of propositions, three $\Delta$-axioms 
\begin{equation}\nonumber
  \Delta(p\supset q)\supset(\Delta p\supset\Delta q), \ 
  \Delta(\Delta p\supset p)\supset\Delta p, \ 
  \Delta p\supset\Delta\Delta p
\end{equation}
and the next three rules of inference: 1) the rule of substitution, 2) the modus ponens rule, and 3) the rule of necessitation which allows to pass from formula $A$ to formula $\Delta A$. 

In the present paper we consider the extension of $GL$, denoted by $GL4$, which can be obtained from $GL$ considering an additional axiom:
\begin{equation}\nonumber
\Delta\Delta 0 \& (\Delta(\Delta p \supset q) \vee (\Delta(\Delta q \supset p)).
\end{equation}

{\bf Magari's algebras.}  
A Magari's algebra \cite{Magari1975a} (also referred to as diagonalizable algebra) $\mathfrak D$ is a boolean algebra $\mathfrak B = (B; \with, \vee, \supset, \neg, \mathbb{0}, \mathbb{1})$ with an additional operator $\Delta$ satisfying the following identities: 
\begin{align}
\Delta(x \supset y)\supset (\Delta x\supset\Delta y)&=\mathbb{1}, \nonumber\\ 
\Delta x \supset \Delta\Delta x &= \mathbb{1},  \nonumber\\ 
\Delta(\Delta x\supset x) &= \Delta x,  \nonumber\\ 
\Delta \mathbb{1} &= \mathbb{1},  \nonumber
\end{align}
where $\mathbb{1}$ is the unit of $\mathfrak B$.  

Interpreting logical connectives of a formula $F$ by corresponding operations on a Magari's algebra $\mathfrak{D}$ we can evaluate any formula of $GL$ on any algebra $\mathfrak{D}$. If for any evaluation of  variables  of $F$ by elements of $\mathfrak{D}$ the resulting value of the formula $F$ on $\mathfrak{D}$ is $\mathbb{1}$ they say $F$ \emph{is valid on $\mathfrak{D}$}. The set of all valid formulas on the given Magari's algebra $\mathfrak{D}$ is an extension of $GL$ \cite{Maximova1989}.    

We consider the 4 valued Magari's algebra $\mathfrak{B}_2 = (\{\mathbb{0},\rho,\sigma,\mathbb{1}\}; \with, \vee, \supset, \neg, \Delta)$, its boolean operations $\with, \vee, \supset, \neg$ are defined as usual, and the operation $\Delta$ is defined as: 
\begin{equation}\nonumber
  \Delta \mathbb{0} = \Delta \rho = \sigma, \  
  \Delta \sigma = \Delta \mathbb{1} = \mathbb{1}.
\end{equation}

{\bf Expressibility of formulas \cite{Kuznetsov1979}.} Suppose in the logic $L$ we can define the equivalence of two formulas. The formula $F$ is said to be (explicitly) expressible via a system of formulas $\Sigma$ in the logic $L$ if $F$ can be obtained from variables and formulas of $\Sigma$ using two rules: a) the rule of weak substitution, which allows to  pass from two formulas, say $A$ and $B$ to the result of substitution of one of them in another in place of any variable $p$ of the formula $\frac{A, B}{A[p/B]}$ (where we denote by $A[p/B]$ the thought substitution); b) if we already get formulas $A$ and we know $A$ is equivalent in $L$ to $B$, then we have also formula $B$. 

{\bf Relations on algebras.} 
They say the formula $F(p_1,\dots,p_n)$ preserves on the Magari's algebra $\mathfrak{D}$ the relation $R(x_1,\dots,x_m)$ if for any elements $\alpha_{11}, \dots, \alpha_{mn}$ of $\mathfrak{D}$ the relations \begin{equation}\nonumber
  R(\alpha_{11},\dots,\alpha_{m1}),\dots,(\alpha_{1n},\dots,\alpha_{mn})
\end{equation}
implies 
\begin{equation}\nonumber
  R(F(\alpha_{11},\dots,\alpha_{1n}),\dots,F(\alpha_{m1},\dots,\alpha_{mn}))
\end{equation}

The relation $R(x_1,\dots,x_m)$ on a finite algebra $\mathfrak{D}$ can be substituted by  a corresponding matrix $\beta_{ik}$ $(i=1,\dots,m,\ k=1,\dots,l)$ of all elements of $\mathfrak{D}$ such that the statement $R(\beta_{1k},\dots,\beta_{mk})$ holds. In this case we speak about preserving of a matrix instead of preserving of a relation on $\mathfrak{D}$.

\section{Preliminary results}
{\bf Representatin of 4-valued operations by formulas.} 
Next theorem gives necessary and sufficient conditions for a 4-valued operation on the set $\{ \mathbb{0},\rho,\sigma,\mathbb{1} \}$ to be expressible via a formula of the propositional provability calculus. 

\begin{thm}\label{th-1-1}
A function $f$ of the general 4-valued logic can be expressed by a formula of the calculus of the logic $L\mathfrak{B}_2$ if and only if it conserves the relation $\Delta x = \Delta y$ on the algebra $\mathfrak{B}_2$.
\end{thm}

\begin{proof}
\textit{Necessity.} It can be easily verified the formulas  $p\& q$, $p\vee q$, $p\supset q$, $\neg p$ \c{s}i $\Delta p$ conserve the relation $\Delta x = \Delta y$ on the algebra $\mathfrak{B}_2$. Since any formula $F$ is directly expressible by them, and, so, the formula $F$ must also preserve the same relation on $\mathfrak{B}_2$. 

\textit{Sufficiency.} Let us to note that to any element of the algebra $\mathfrak{B}_2$ corresponds a constant of the logic $L\mathfrak{B}_2$, so, in the sequel we denote the elements of the algebra $\mathfrak{B}_2$ and the constants of the logic $L\mathfrak{B}_2$ by the same symbols. Suppose the operation $f(p_1,\dots,p_n)$ conserves the relation $\Delta x = \Delta y$ on the algebra $\mathfrak{B}_2$. We will show in the following how to design the formula  $F(p_1,\dots,p_n$), which represent the operation $f$ on the algebra $\mathfrak{B}_2$. 

Examine an arbitrary fixed set $\alpha=(\alpha_1,\dots,\alpha_n)$ of elements of $\mathfrak{B}_2$. Let $f(\alpha_1,\dots,\alpha_n)=\delta$ and consider the formula  $(\&_{i=1}^{n}\square(p_i\sim\alpha_i))\&\delta$ denoted by $C^\alpha(p_1,\dots,p_n)$. It can be verified that $C^\alpha$ satisfies the following conditions: 
\[
  C^\alpha(p_1,\dots,p_n)=
  \begin{cases}
    \delta, & \text{if $p_i=\alpha_i$, $i=1,\dots,n$}\\
    \square 0\&\sigma, & \text{if $\forall i: \Delta p_i
    =\Delta\alpha_i$, \c{s}i $\exists i: p_i \not=\alpha_i, $}\\
    0,      & \text{if $\exists j: \Delta p_j
    \not=\Delta\alpha_j$}
  \end{cases}.
\]
Denote with $\Gamma$ the set of all ordered sets of 4 elements from the set $\{  \mathbb{0},\rho,\sigma, \mathbb{1} \}$. Consider the formula 
\begin{equation}\label{eq-1-7}
  F(p_1,\dots,p_n)=\bigvee_{\gamma\in\Gamma} C^\gamma(p_1,\dots,p_n)
\end{equation}
Let us show the formula $F$ is the thought for one. To prove this it is sufficient  to convince ourselves that $F[\alpha_1,\dots,\alpha_n]$ $=$ $f(\alpha_1$, $\dots$, $\alpha_n)$ since the set of elements $\alpha$ is taken arbitrarily. The relation (\ref{eq-1-7}) can be rewritten as:  
\begin{equation}\label{eq:1-7a}
	\begin{split}
  F(p_1,\dots,p_n) = \bigvee_{\gamma\in\Gamma, \gamma=\alpha}
  &C^\gamma(p_1,\dots,p_n)\vee \\
  \bigvee_{\gamma\in\Gamma, \exists i: \Delta\gamma_i\not=\Delta\alpha_i}
  &C^\gamma(p_1,\dots,p_n)\vee \\
  \bigvee_{\alpha\not=\gamma, \Delta\gamma_i=\Delta\alpha_i}
  &C^\gamma(p_1,\dots,p_n).
  \end{split}
\end{equation}
The last relation (\ref{eq:1-7a}) implies, taking into consideration the properties of the formula $C^\alpha$, the following equality: 
\[
 \begin{aligned}
  F[\alpha_1,\dots,\alpha_n]=
    & C^\alpha(\alpha_1,\dots,\alpha_n) \vee \\
    & \bigvee_{\gamma\in\Gamma, \exists i: \Delta\gamma_i\not=\Delta\alpha_i}
      C^\gamma(\alpha_1,\dots,\alpha_n)
      \vee \\
    & \bigvee_{\alpha\not=\gamma, \Delta\gamma_i=\Delta\alpha_i}
      C^\gamma(\alpha_1,\dots,\alpha_n)= \\
    & \delta \vee 0 \vee (\square\delta\&\sigma)=
      \delta =
      f(\alpha_1,\dots,\alpha_n).
 \end{aligned}
\]
Hence, for an arbitrary set of elements $\alpha\in\Gamma$ we have  
\[
  F[\alpha_1,\dots,\alpha_n]=f(\alpha_1,\dots,\alpha_n).
\]
So, the formula $F$ realizes the operation $f$ on the algebra $\mathfrak{B}_2$.

The theorem \ref{th-1-1} is proved. 
\end{proof}

The next statement is a consequence of the above theorem. 

\begin{prop}
There are 64 unary formulas in the calculus of the logic $L\mathfrak{B}_2$ which are not equivalent each other in $L\mathfrak{B}_2$ and realize the corresponding unary operations of the algebra $\mathfrak{B}_2$. 
\end{prop}

\begin{table}[h]
    \caption{Unary operations of  $\mathfrak{B}_2$}\label{tabel:1-1}
    \centering
    \begin{tabular}{|c|c|c|c|c|c|c|c|c|}
        \hline
        $p$ & $I_{1j}$ & $I_{2j}$ & $I_{3j}$ & $I_{4j}$ & $I_{5j}$ & $I_{6j}$ & $I_{7j}$ & $I_{8j}$ \\
        \hline
        $0$         & $0$ & $0$ & $\rho$   & $\rho$ & $\sigma$ & $\sigma$ & $1$ & $1$ \\
        $\rho$      & $0$ & $ \rho$ & $0$ & $\rho$ & $\sigma$ & $1$ & $\sigma$ & $1$ \\
        \hline
        $p$ & $I_{i1}$ & $I_{i2}$ & $I_{i3}$ & $I_{i4}$ & $I_{i5}$ & $I_{i6}$ & $I_{i7}$ & $I_{i8}$ \\
        \hline
        $\sigma$    & $0$ & $0$ & $\rho$   & $\rho$ & $\sigma$ & $\sigma$ & $1$ & $1$ \\
        $1$         & $0$ & $ \rho$ & $0$ & $\rho$ & $\sigma$ & $1$ & $\sigma$ & $1$ \\
        \hline
    \end{tabular}
\end{table}

In order to describe the derived unary operations of the algebra $\mathfrak{B}_2$ we use  the table \ref{tabel:1-1}, where $I_{ij}(p)$ $(i=1,\dots,8; \ j=1,\dots,8)$ denotes the unary operation which for $p=0$ and $p=\rho$ takes values from the $i$-th column, and for $p=\sigma$ and $p=1$ it takes values from the $j$-th column.

For example, $I_{11}=0$, $I_{16}=p$, $I_{73}=\neg p$, $I_{58}=\Delta p$, $I_{88}=1$.

%% 3-2-exprimab-const.tex

\renewcommand{\baselinestretch}{1.0}\normalsize%
\section{Main result}\label{sec:expr-const}
\renewcommand{\baselinestretch}{1.33}\normalsize%

Consider the following relations on $\mathfrak{B}_2$ (read symbols "$==$" as "defined by"): 

1) $R_1(x)==(\Delta x=\sigma)$;

2) $R_2(x)==(\Delta x=1)$;

3) $R_3(x)==I_{15}(x)=x)$;

4) $R_4(x)==I_{18}(x)=x)$;

5) $R_5(x)==I_{45}(x)=x)$;

6) $R_6(x)==I_{48}(x)=x)$;

7) $R_7(x)==I_{25}(x)=x)$;

8) $R_8(x)==I_{28}(x)=x)$;

9) $R_9(x)==I_{16}(x)=x)$;

10) $R_{10}(x)==I_{46}(x)=x)$;

11) $R_{11}(x,y)==(I_{37}(x)=y)$;

12) $R_{12}(x,y)==(\Delta x\not=\Delta y)$;

We denote by $\mathfrak{M}_i$ the corresponding matrix to the relation $R_i$ on the  algebra $\mathfrak{B}_2$ and denote with $\Pi_i$ the class of all formulas, which preserves the relation $R_i$ on the algebra $\mathfrak{B}_2$, i.e. the class of all formulas, which conserves the matrix $\mathfrak{M}_i$ on $\mathfrak{B}_2$ for any  $i=1,\dots,12$. 

The table  \ref{tabel:21} presents the list of all classes $\Pi_1,\dots,\Pi_{12}$ and their corresponding matrix. 

%
%
%
%

%%
%% Tabelul cu matricele corespunzatoare predicatelor
%%
%%
%% 3-1-tabel-1,matricele-1-12.tex

%\setlength{\extrarowheight}{6pt}
\begin{longtable}{|c|c|}
    \caption{The class of formulas and the corresponding matrix}\label{tabel:21} \\ 
    		\hline
        The class   & Defining matirx \\ 
    		\hline
        \endfirsthead 
        \multicolumn{2}{r}{\small\itshape The next part of the table \ref{tabel:21}} \\ 
    		\hline
        The class & Defining matrix \\ 
    		\hline
        \endhead 
    		\hline
        \multicolumn{2}{r}{\small\itshape The table \ref{tabel:21} continues on the next page} \\ 
    		%\hline
        \endfoot 
    		\hline
        \endlastfoot
        $\Pi_1$ & $\left(0\rho\right)$ \\
        \hline
        $\Pi_2$ & $\left(\sigma 1\right)$ \\
        \hline
        $\Pi_3$ & $\left(0 \sigma\right)$ \\
        \hline
        $\Pi_4$ & $\left(0 1\right)$ \\
        \hline
        $\Pi_5$ & $\left(\rho\sigma\right)$ \\
        \hline
        $\Pi_6$ & $\left(\rho 1\right)$ \\
        \hline
        $\Pi_7$ & $\left(0\rho\sigma\right)$ \\
        \hline
        $\Pi_8$ & $\left(0\rho 1\right)$ \\
        \hline
        $\Pi_9$ & $\left(0\sigma 1\right)$ \\
        \hline
        $\Pi_{10}$ & $\left(\rho\sigma 1\right)$ \\
        \hline
        $\Pi_{11}$ &
                    \rule[-3mm]{0mm}{9mm}
                    $\left(
											\ccaa{0}{\rho}
											\ccaa{\rho}{0}
											\ccaa{\sigma}{1}
											\ccaa{1}{\sigma}
                    \right)$ \\
        \hline
        $\Pi_{12}$ &
                    \rule[-3mm]{0mm}{9mm}
                    $\left(
											\ccaa{0}{\sigma}
											\ccaa{0}{1}
											\ccaa{\rho}{\sigma}
											\ccaa{\rho}{1}
											\ccaa{\sigma}{0}
											\ccaa{\sigma}{\rho}
											\ccaa{1}{0}
											\ccaa{1}{\rho}
%                        \begin{array}{c}
%                            0 0 \rho \rho \sigma \sigma 1 1 \\
%                            \sigma 1 \sigma 1 0 \rho 0 \rho
%                        \end{array}
                    \right)$ \\
        \hline
\end{longtable}

\begin{thm}\label{lema2-1}
Suppose the formulas $F_1, \dots, F_{12}$ do not preserve the corresponding relations $R_{1},\dots,R_{12}$ on the Magari's algebra $\mathfrak{B}_2$. The constants $\mathbb{0},\rho,\sigma,\mathbb{1}$ are expressible in the logic $L\mathfrak{B}_2$ via formulas $F_1, \dots, F_{12}$.
\end{thm}

The proof of the theorem follows from the next 5 lemmas. 
 
\begin{ilem}\label{lema2-1-1}
The formula $A(p)$, where 
  \begin{equation}\label{eq:2-2}
    A[0]\in\{\sigma,1\}
  \end{equation}
  is expressible $L\mathfrak{B}_2$ via formula $F_1$.
\end{ilem}

\begin{proof}
Really, the formula $F_1$ does not conserve the relation $R_1$ on the algebra $\mathfrak{B}_2$. Then there exists an ordered set of elements  $\seq{\alpha_1,\dots,\alpha_n}$ from $\mathfrak{B}_2$ such that  
\begin{eqnarray}
    &\alpha_i\in\set{0,\rho} \quad (i=1,\dots,n)\label{eq:2-3} \\
    &F_1[\alpha_1,\dots,\alpha_n]\in\set{\sigma,1}\label{eq:2-4}
\end{eqnarray}
Since $F_1$ conserves the predicate $\Delta x = \Delta y$ on the algebra 
$\mbb$, in view of relations  (\ref{eq:2-3}) \c{s}i (\ref{eq:2-4}) we also have that  
\begin{equation}
    F_1[0,\dots,0]\in\set{\sigma,1}\label{eq:2-5}
\end{equation}
Let $A(p)=F_1[p_1/p,\dots,p_n/p]$. In virtue of (\ref{eq:2-5}) we obtain $A[0]\in\set{\sigma,1}$. 
\end{proof}

\begin{ilem}\label{lema2-1-2}
The formula $B(p)$, where  
  \begin{equation}\label{eq:2-6}
    B[1]\in\{0,\rho\}
  \end{equation}
  is expressible in $L\mathfrak{B}_2$ via $F_2$.
\end{ilem}

\begin{proof}
The validity of lemma \ref{lema2-1-2} follows from the fact that its formulation is dualistic to the formulation of lemma  \ref{lema2-1-1} with respect to $\neg p$, where formula $B$ is considered in place of the corresponding formula $A$. 
\end{proof}

\begin{ilem}\label{lema2-1-3}
Let the formulas $A$ and $B$ satisfy the relations (\ref{eq:2-2}), (\ref{eq:2-6}) and  
    \begin{equation}\label{eq:2-7}
        B[\sigma]=B[1].
    \end{equation}
Then at one of the constants  $0$ or $\rho$ is expressible via formulas $A, B$ and $F_{12}$ in the logic  \lmbb.
\end{ilem}

\begin{proof}
Let $B$ satisfies the relations (\ref{eq:2-6}) and (\ref{eq:2-7}). Then two cases are possible for the formula $B$: 1) $B[0]\in\set{0,\rho}$; 2) $B[0]\in\set{\sigma,1}$. Let us observe in the first case the formula $B[A[B(p)]]$ is equivalent to one of the constants $0$ or $\rho$. 

Consider case 2), i.e. $B[0]\in\set{\sigma,1}$. Consider formula $F_{12}$, which does not preserve $R_{12}$ on  \mbb. Then there are exist two ordered sets of elements  \seq{\alpha_1,\dots,\alpha_n} and  \seq{\beta_1,\dots,\beta_n} from the algebra $\mbb$ such that  
\begin{eqnarray}
    &\Delta\alpha_i\not=\Delta\beta_i\quad(i=1,\dots,n)\label{eq:2-8}&{} \\
    &\Delta F_{12}[\alpha_1,\dots,\alpha_n] =
    \Delta F_{12}[\beta_1,\dots,\beta_n]&{}\label{eq:2-9} 
\end{eqnarray}
We build the formula $D(p_1,\dots,p_8)=F_{12}[D_1,\dots,D_n]$, where for every $i=1,\dots,n$
\[
D_i(p_1,\dots,p_8) = p_1, \mbox{if } \alpha_i=0, \beta_i=\sigma,
\]\[
D_i(p_1,\dots,p_8) = p_2, \mbox{if } \alpha_i=0, \beta_i=1,
\]\[
D_i(p_1,\dots,p_8) = p_3,  \mbox{if } \alpha_i=\rho, \beta_i=\sigma,
\]\[
D_i(p_1,\dots,p_8) = p_4,  \mbox{if } \alpha_i=\rho, \beta_i=1,
\]\[
D_i(p_1,\dots,p_8) = p_5, \mbox{if } \alpha_i=\sigma, \beta_i=0,
\]\[
D_i(p_1,\dots,p_8) = p_6, \mbox{if } \alpha_i=\sigma, \beta_i=\rho,
\]\[
D_i(p_1,\dots,p_8) = p_7, \mbox{if } \alpha_i=1, \beta_i=0,
\]\[
D_i(p_1,\dots,p_8) = p_8, \mbox{if } \alpha_i=1, \beta_i=\rho
\]
%\[
%    D_i(p_1,\dots,p_8)=
%    \left\{
%        \begin{array}{rl}
%            p_1, & \mbox{if } \alpha_i=0, \beta_i=\sigma,\\
%            p_2, & \mbox{if } \alpha_i=0, \beta_i=1,\\
%            p_3, & \mbox{if } \alpha_i=\rho, \beta_i=\sigma,\\
%            p_4, & \mbox{if } \alpha_i=\rho, \beta_i=1,\\
%            p_5, & \mbox{if } \alpha_i=\sigma, \beta_i=0,\\
%            p_6, & \mbox{if } \alpha_i=\sigma, \beta_i=\rho,\\
%            p_7, & \mbox{if } \alpha_i=1, \beta_i=0,\\
%            p_8, & \mbox{if } \alpha_i=1, \beta_i=\rho
%        \end{array}
%    \right.
%\]
(by the power of relation (\ref{eq:2-8}) other cases are impossible). It is clear that $D_i[0$, $0$, $\rho$, $\rho$, $\sigma$, $\sigma$, $1$, $1]=\alpha_i$ and $D_i[\sigma,1,\sigma,1,0,\rho,0,\rho]=\beta_i$. Then, taking into account the design of the formula $D$, the relation (\ref{eq:2-9}) and the last equalities, we obtain 
\begin{equation}\label{eq:2-10}
    \left(
        \begin{array}{c}
            D[0,0,\rho,\rho,\sigma,\sigma,1,1]\\
            D[\sigma,1,\sigma,1,0,\rho,0,\rho]
        \end{array}
    \right)
    \subseteq
    \left(
        \begin{array}{c}
            0 0 \rho \rho \sigma \sigma 1 1\\
            0 \rho 0 \rho \sigma 1 \sigma 1
        \end{array}
    \right)
\end{equation}
Consider now the formula $D^*(p,q)=D[p,p,p,p,q,q,q,q]$. By (\ref{eq:2-10}) and the fact that  $D$ conserves on the algebra  $\mbb$ the predicate $\Delta x = \Delta y$, we obtain
\begin{equation}\label{eq:2-11}
    \left(
        \begin{array}{c}
            D^*[0,1]\\
            D^*[1,0]
        \end{array}
    \right)
    \subseteq
    \left(
        \begin{array}{c}
            0 0 \rho \rho \sigma \sigma 1 1\\
            0 \rho 0 \rho \sigma 1 \sigma 1
        \end{array}
    \right)
\end{equation}
Let us examine the formula $D'(p,q)$, defined by the scheme 
\[
    D'(p,q)=
    \left\{
        \begin{array}{ll}
            D^*(p,q), & \mbox{if } D^*[0,1]\in\set{0,\rho},\\
            B[D^*(p,q)], & \mbox{if } D^*[0,1]\in\set{\sigma,1}.
        \end{array}
    \right.
\]
By power of the relation (\ref{eq:2-7}) and taking into consideration  (\ref{eq:2-11}), the formula $D'$ satisfies the inclusion $\set{D'[0,1],D'[1,0]}\subseteq\set{0,\rho}$. Therefore, in the second case, taking into consideration  (\ref{eq:2-6}), the relation $B[D'[p,B(p)]]]\in\set{\sigma,1}$ holds. Hence, on the basis of the conditions (\ref{eq:2-6}) and (\ref{eq:2-7}), the formula $B[B[D'[p$, $B(p)]]]$ is equivalent  to one of the constants  $0$ or $\rho$. 
\end{proof}

\begin{ilem}\label{lema2-1-4}
Let formulas $A$ and $B$ satisfy the relations (\ref{eq:2-2}), (\ref{eq:2-6}) and  
    \begin{equation}\label{eq:2-12}
        B[\sigma]\not=B[1].
    \end{equation}
Then at least one of the constants  $0$ or $\rho$ is expressible via formulas  $A$, $B$, $F_{3},F_{7},F_{11},F_{12}$ in the logic \lmbb.
\end{ilem}

\begin{proof}
The relation (\ref{eq:2-6}) and the fact that $B$ conserves the predicate $\Delta{x}=\Delta{y}$ on the algebra $\mbb$ implies that there are two possible situations: 1) $B[1]=\rho$, and 2) $B[1]=0$. 

Let us consider the first case. On the basis of the relation (\ref{eq:2-12}) we have that $B[\sigma]=0$. We consider the formula $F_3$. Since it does not conserve $R_3$ on \mbb there exists an ordered set \set{\alpha_1,\dots,\alpha_n} of elements of $\mbb$ such that  
\begin{eqnarray}
    {} & \alpha_i\in\set{0,\sigma},\quad i=1,\dots,n\label{eq:2-13}\\
    {} & F_{3}[\alpha_1,\dots,\alpha_n]
    \in\set{\rho,1}.\label{eq:2-14}
\end{eqnarray}
We design the formula $E(p)=F_3[\sta{E}{1}{n}]$, where for every $i=\stb{1}{n}$
\[
    E_i(p)=
    \left\{
        \begin{array}{ll}
            B(p), & \mbox{if } \alpha_i=0,\\
            p, & \mbox{if } \alpha_i=\sigma
        \end{array}
    \right.
\]
(in accordance to (\ref{eq:2-13}) other cases are impossible for the elements $\alpha_i$). The formula $E$ is direct expressible via formulas $F_3$ and $B$. Obviously, $E_i[\sigma]=\alpha_i$ and the view of relation  (\ref{eq:2-14}) we have 
$E[\sigma]=F_3[E_{1}[\sigma],\dots$, $E_{n}[\sigma]] = F_3[\alpha_{1},\dots,\alpha_{n}]
\in \set{\rho,1}$. 
Consider the formula $E^*(p)$, defined by the scheme
\[
    E^*(p)=
    \left\{
        \begin{array}{ll}
            E(p), & \mbox{if } E[\sigma]=\rho,\\
            B[E(p)], & \mbox{if } E[\sigma]=1.
        \end{array}
    \right.
\]
The formula $E^*(p)$ is directly expressible via formulas $B$ and  $E$ and satisfies the condition 
\begin{equation}\label{eq:2-15}
    E^*[\sigma]=\rho.
\end{equation}
Two sub-cases are possible: 1.1) $E^*[1]=\rho$ and 1.2) $E^*[1]=0$. In the sub-case  1.1) the formula $E^*(p)$ satisfies analogous conditions to  (\ref{eq:2-6}) for the formula $B(p)$ from lemma \ref{lema2-1-3} and then the proof will follow the corresponding proof of the lemma  \ref{lema2-1-3}, thus one of two constants $0$ or $\rho$ is obtained. 

Consider now the sub-case 1.2) when $E^*[1]=0$. Consider formula $F_7$. Since $F_7$ does not conserve the relation  $R_7$ on  \mbb, then there exists an ordered set of elements  \seq{\sta{\beta}{1}{n}} from  $\mbb$ such that 
\begin{eqnarray}
  &\beta_i\in\set{0,\rho,\sigma},\quad i=1,\dots,n\label{eq:2-16}\\
  &F_7[\sta{\beta}{1}{n}]=1\label{eq:2-17}
\end{eqnarray}
Take the formula $H(p)=F_7[\sta{H}{1}{n}]$, where for every $i=1,\dots,n$
\[
					H_i(p) = B(p),  \mbox{if } \beta_i=0,
\]\[
          H_i(p) =  E^*(p),  \mbox{if } \beta_i=\rho,
\]\[
          H_i(p) =   p,  \mbox{if } \beta_i=\sigma.
\]
%\[
%    H_i(p)=
%    \left\{
%        \begin{array}{ll}
%						B(p), & \mbox{if } \beta_i=0,\\
%            E^*(p), & \mbox{if } \beta_i=\rho,\\
%            p, & \mbox{if } \beta_i=\sigma.
%        \end{array}
%    \right.
%\]
(obviously, other cases are missed for the elements $\beta_i$). The formula $H$ is directly expressible via $F_7,B$ and $E^*$. It is clear $H_i[\sigma]=\beta_i$ and in agreement with relation  (\ref{eq:2-17}) we have  
\begin{equation}\label{eq:2-18}
    H[\sigma]=1.
\end{equation}

If $H[1]=1$ then the formula $B[H(p)]$ satisfies analogous conditions  to conditions (\ref{eq:2-6}) and  (\ref{eq:2-7}) for the formula $B$ from lemma \ref{lema2-1-3}. That is why we can obtain one of the constants $0$ or $\rho$ in the case when $H[1]=1$ in the same way as in lemma \ref{lema2-1-3}. 

Let $H[1]=\sigma$. Use the formula $F_{11}$. It follows from its properties that there exist two ordered sets of elements $(\sta{\gamma}{1}{n})$ and $(\sta{\delta}{1}{n})$ from \mbb, such that the next relation holds
\begin{equation}\label{eq:2-19}
    I_{37}[\gamma_i]=\delta_i, \quad i=1,\dots,n
\end{equation}
Taking also into consideration the theorem \ref{th-1-1} we have 
\begin{equation}\label{eq:2-20}
    F_{11}[\sta{\gamma}{1}{n}]=
    F_{11}[\sta{\delta}{1}{n}].
\end{equation}
Design the formula $J(p)=F_{11}[\stc{J}{1}{n}{(p)}]$, where for any $i=1,\dots,n$ we get 
\[
    J_i(p)=
    \left\{
        \begin{array}{ll}
            B(p), & \mbox{if } \gamma_i=0,\delta_i=\rho,\\
            E^*(p), & \mbox{if } \gamma_i=\rho,\delta_i=0,\\
            p, & \mbox{if } \gamma_i=\sigma,\delta_i=1,\\
            B(p), & \mbox{if } \gamma_i=1,\delta_i=\sigma
        \end{array}
    \right.
\]
(by properties of the relation (\ref{eq:2-19}) the elements $\gamma_i$ and $\delta_i$ do not take other values). $J(p)$ is directly expressible via $B$, $E^*$, $H$ and  $F_{11}$. Let us notice that $J_i[\sigma]=\gamma_i$, $J_i[1]=\delta_i$ \c{s}i, hence, by relation (\ref{eq:2-20}), we obtain $J[\sigma]=J[1]$. So, the formula $J^*(p)$, defined by the scheme 
\[
    J^*(p)=
    \left\{
        \begin{array}{ll}
            J(p), & \mbox{if } J[1]\in\set{0,\rho},\\
            B[J(p)], & \mbox{if } J[1]\in\set{\sigma,1},
        \end{array}
    \right.
\]
satisfies the relations 
\begin{equation}\label{eq:2-21}
    J[\sigma]=J[1],\; J[1]\in\set{0,\rho}
\end{equation}
Let us notice that conditions (\ref{eq:2-21}) are analogous to conditions  (\ref{eq:2-6}) and (\ref{eq:2-7}) from lemma \ref{lema2-1-3}. Hence, we can obtain in a similar way one of the constants  $0$ or $\rho$. So, the proof of the lemma (\ref{lema2-1-4}) in the case  1) is finished. 

Let us consider the second case, when  $B[1]=0$. Examine the formula $F_4$. Since it does not conserve the relation $R_4$ on \mbb, then there is an ordered set of elements $(\sta{\varepsilon}{1}{n})$ on \mbb{} such that 
\begin{eqnarray}
  &\varepsilon_i\in\set{0,1},\quad i=1,\dots,n\label{eq:2-22}\\
  &F_4[\sta{\varepsilon}{1}{n}]\in\set{\rho,\sigma}\label{eq:2-23}
\end{eqnarray}
Design the formula $S(p)=F_4[\sta{S}{1}{n}]$, where for any $i=1,\dots,n$ we have 
\[
    S_i(p)=
    \left\{
        \begin{array}{ll}
            B(p), & \mbox{if } \varepsilon_i=0,\\
            p, & \mbox{if } \varepsilon_i=1.
        \end{array}
    \right.
\]
(by properties of (\ref{eq:2-22}) we do not have other cases). The formula $S$ is directly expressible via $F_4$ \c{s}i $B$. Obviously $S_i[1]=\varepsilon_i$ and in agreement with relation (\ref{eq:2-23}) we have 
\[
    S[1]=F_4[\stc{S}{1}{n}{[1]}]=
    F_4[\sta{\varepsilon}{1}{n}]\in\set{\rho,\sigma}.
\]
Consider the formula $S^*(p)$ defined by the scheme 
\[
    S^*(p)=
    \left\{
        \begin{array}{ll}
            S(p), & \mbox{if } S[1]=\rho,\\
            B[S(p)], & \mbox{if } S[1]=\sigma.
        \end{array}
    \right.
\]
The formula $S^*(p)$  is directly expressible via  $B$ and  $S$ and verifies the condition $S^*[1]=\rho$. taking into consideration the theorem \ref{th-1-1} we also have the relation $S^*[\sigma]\in\set{0,\rho}$. If  $S^*[\sigma]=\rho$, then we obtain one of the constants $0$ or $\rho$ as in the case 1). It remains to consider the case when $S^*[\sigma]=0$. But in this case we are already under conditions of the first case, which was successfully considered already. 
\end{proof}

\begin{ilem}\label{lema2-1-5}
All constants $0,\rho,\sigma,1$ are expressible in the logic $\lmbb$ via formulas $F_i,\; i=3,\dots,10$, via any unary formulas $A$ and $B$, which verify the corresponding conditions (\ref{eq:2-2}) and (\ref{eq:2-6}), and via any constant $0$ or $\rho$. 
\end{ilem}

\begin{proof}
Let us convince ourselves that one of the following systems of formulas (\ref{eq:2-24}) is expressible via one of the constants $0$ or $\rho$ and the formula $A$: 
\begin{equation}\label{eq:2-24}
    \set{0,\sigma},\;
    \set{0,1},\;
    \set{\rho,\sigma},\;
    \set{\rho,1}.
\end{equation}

Let us consider we have the  constant $0$. By properties of  (\ref{eq:2-2}) we have $A[0]\in\set{\sigma,1}$, which means we have at least one of the first two systems of (\ref{eq:2-24}). Suppose we have the constant  $\rho$. Then by theorem  \ref{th-1-1} we have $A[\rho]\in\set{\sigma,1}$, and by the similar reasons as in the case of the constant $0$ we can conclude analogously we have at least one the the last two systems of  (\ref{eq:2-24}). 

We wil show in the following that via every system of formulas of the list (\ref{eq:2-24}) and via corresponding formulas  $F_3$, $F_4$, $F_5$, $F_6$ is expressible one of the following systems of constants
\begin{equation}\label{eq:2-25}
    \set{0,\rho,\sigma},\;
    \set{0,\rho,1},\;
    \set{0,\sigma,1},\;
    \set{\rho,\sigma,1}.
\end{equation}

Let us consider the system of formulas $\set{0,\sigma}$. Examine the formula $F_3$. Obviously via $F_3$ and constants $0$ and $\sigma$ is expressible some formula  $F^*_3(p,q)$, which satisfies the condition $F^*_3[0,\sigma]\in\set{\rho,1}$. Hence, we obtain one of the systems of formulas \set{0,\rho,\sigma} or \set{0,\sigma,1}. In a similar way we obtain: 
\begin{itemize}
\item the system \set{0,\rho,1} or the system \set{0,\sigma,1} via \set{0,1} and $F_4$; 
\item the system \set{0,\rho,\sigma} or the system \set{\rho,\sigma,1} via \set{\rho,\sigma} and $F_5$; 
\item the system \set{0,\rho,1} or the system \set{\rho,\sigma,1} via \set{\rho,1} and $F_6$. 
\end{itemize}

In a similar manner we obtain that all constants of the system \set{0,\rho,\sigma,1} are  expressible in  \lmbb{} via every system of formulas of (\ref{eq:2-25}) and corresponding formulas $F_7,F_8,F_9,F_{10}$. 
\end{proof}

\section{Conclusions}

Theorem \ref{lema2-1} provide us only sufficient conditions for expressibility of constants of the propositional provability logic $L\mathfrak{B}_2$. We can consider a slice of extensions of $GL$ \cite{Blok1980}, which also has an additional axiom $\Delta\Delta p$ and examine the conditions of expressibility of constants in these logics too. Note the logic $L\mathfrak{B}_2$ is an element of this slice of extensions. Also we can examine other types of expressibility of formulas: implicit expressibility, parametric expressibility, existential expressibility, etc.  

\renewcommand{\baselinestretch}{1.0}\normalsize%

\end{document}